\documentclass[12pt,a4paper]{article}

\usepackage[utf8]{inputenc}
\usepackage[english]{babel}
\usepackage{graphicx}
\usepackage{multirow}
\usepackage{caption}
\usepackage{cancel}
\usepackage{bm}
\usepackage{url}

\usepackage{amssymb,amsmath,amsthm}
\usepackage{vmargin}
\usepackage{natbib}

\usepackage[usenames, dvipsnames]{color}
\usepackage{dsfont}
\usepackage{algorithm}
\usepackage{algpseudocode}
\usepackage{pifont}
\usepackage{algcompatible}
\usepackage{caption}
\usepackage{subcaption}

\newtheorem{prop}{Proposition}
\newtheorem{theorem}{Theorem}

\newcommand{\R}{\mathbb R}

\newcommand{\E}{\mathbb E}

\renewcommand{\P}{\mathbb P}
\renewcommand{\E}{\mathbb E}

\newcommand{\teta}{\mbox{\boldmath{$\theta$}}}
\newcommand{\Beta}{\mbox{\boldmath{$\beta$}}}
\newcommand{\gr}[1]{\textbf{#1}}

\setpapersize{A4}
\setmargins{2.5cm}{2cm}{16cm}{25cm}{0cm}{0cm}{1cm}{1cm}

\usepackage{sectsty}
\sectionfont{\large}

\providecommand{\keywords}[1]{\textit{Keywords:} #1}

\begin{document}

\title{A note on perfect simulation for exponential random graph models}

\author{A.~Cerqueira\footnote{Instituto de Matem\'atica e Estat\'\i stica, Universidade de S\~ao Paulo, Brazil}, A.~Garivier\footnote{Institut de Math\'ematiques de Toulouse, Universit\'e Paul Sabatier, Toulouse, France} and F.~Leonardi\footnotemark[1]}

\date{\today}

\maketitle

\begin{abstract}
In this paper we propose a perfect simulation algorithm for the Exponential Random Graph Model, based on the
Coupling From The Past method of \citet{propp1996exact}. We use a  Glauber dynamics to construct the Markov Chain and we prove the monotonicity of the ERGM for a subset of the parametric space. We also obtain an upper bound on the running time of the algorithm that depends on the mixing time of the Markov chain. 
\end{abstract}


\section{Introduction}
\label{sec:introdution}

In recent years there has been an increasing interest in the study of probabilistic models defined on graphs. One of the main reasons for this interest is the flexibility of these models, making them suitable for the description of real datasets, for instance in social networks \citep{newman2002random} and neural networks \citep{cerqueira2017test}. The interactions in a real network can be represented through a graph, where the interacting objects are represented by the vertices of the graph and the interactions between these objects are identified with the edges of the graph.

In this context,  the Exponential Random Graph Model (ERGM) has been vastly explored to model real networks \citep{robins2007introduction}. Inspired on the Gibbs distributions of Statistical Physics, one of the reasons of the popularity of ERGM can be attributed to the specific definition of the probability distribution that takes into account 
appealing properties of the graphs such as number of edges, triangles, stars, etc.  

In terms of statistical inference, learning the parameters of a ERGM is not a simple task. 
Classical inferential methods such as Monte Carlo maximum likelihood \citep{geyer1992constrained} or  pseudolikelihood  \citep{strauss1990pseudolikelihood} have been proposed, but exact calculation of the estimators is computationally 
infeasible except for small graphs. 
In the general case, estimation is made by means of a Markov Chain Monte Carlo (MCMC) procedure  \citep{snijders2002markov}. It consists in building a Markov chain whose stationary distribution is the target law. The classical approach is to simulate a ``long enough'' random path, so as to get an approximate sample. But this procedure suffers from a  ``burn-in'' time that is difficult to estimate in practice  and  can be exponentially long, as shown in \citet{bhamidi2011mixing}.
To overcome this difficulty, \citet{butts2015novel} proposed a novel algorithm called \emph{bound sampler}. The simulation time of the bound sampler is fixed and depends only on the number of vertices of the graph. Despite being approximate methods for sampling from the ERGM, the bound sampler is recommended when the mixing time of the chain is long  \citep{butts2015novel}.

In order to get a sample of the stationary distribution of the chain, an alternative approach introduced by \citet{propp1996exact} is  called \emph{Coupling From The Past} (CFTP). This method yields a value distributed under the \emph{exact} target distribution, and is thus called perfect (also known as exact) simulation. Like MCMC, it does not require computing explicitly the normalising constant; but it does not require any burn-in phase either, nor any detection of convergence to the target distribution.

In this work we present a  perfect simulation algorithm for the ERGM. In our approach, we adapt a version of  the CFTP algorithm, which was originally developed for spin systems \citep{propp1996exact} to construct a perfect simulation algorithm  using a Markov chain based on   the Glauber dynamics.  For a specific family of distributions 
based on simple statistics we prove that they satisfy the monotonicity property and then the CFTP algorithm is computationally efficient to generate a sample from the target distribution. Finally, and most importantly, we prove that the mean running time of the perfect simulation algorithm is at most as large (up to logarithmic terms) than the mixing time of the Markov chain. As a consequence, we argue CFTP is superior to  MCMC for the simulation of the ERGM.

\section{Definitions}
\label{sec:definitions}

\subsection{Exponential random graph model}
\label{sec:erg}

A graph is a pair $\gr{g}=(V,E)$, where $V=\{1,2,\dots,N\}$ is a finite set of vertices and $E\subset V\times V$ is a set of edges that connect pairs of vertices. In this work we consider undirected graphs; i.e, graphs for which the set $E$ satisfies the following properties:

\begin{enumerate}
\item For all $i \in V$, $(i,i) \notin E$
\item If $(i,j) \in E$, then $(j,i) \in E$.
\end{enumerate}

 Let $\mathcal{G}_N$ be the collection of all undirected graphs with set of vertices V.
A graph $(V,E) \in \mathcal{G}_N$ can be identified with a symmetric binary  matrix $\gr{x}\in\mathcal{M}_N\big(\{0,1\}\big)$ such that $x(i,j)=1$ if and only if $(i,j)\in E$.
 
We will thus denote by $\gr{x}$ a graph, which is completely determined by the values  $\gr{x}(i,j)$, for $1 \leq i < j \leq N$.
The \emph{empty graph} $\gr{x}^{(0)}$ is given by $x^{(0)}(i,j)=0$, for $1 \leq i < j \leq N$,  and the  the \emph{complete graph} $\gr{x}^{(1)}$ is given by $x^{(1)}(i,j)=1$, for $1 \leq i < j \leq N$.

  We write $\gr{x}_{-ij}$ to represent the set of values 
  $\{\gr{x}(l,k)\colon (i,j)\neq (l,k)\}$. 
 For a given graph $\gr{x}$ and a set of vertices $ U \subseteq V$, we denote by $\gr{x}(U)$ the subgraph induced by $U$, that is $\gr{x}(U)$ is a graph with set of vertices $U$ and such that $\gr{x}(U)(i,j)=\gr{x}(i,j)$ for all $i,j \in U$.
 
 In the ERGM model, the probability of selecting a graph in $\mathcal{G}_N$ is defined by its structure: its number of edges, of triangles, etc.
  
To define subgraphs counts in a graph, let $V_m$ be the set of all possible permutations of $m$ distinct elements of $V$. For $v_m \in V_m$, define $\bold{x}(v_m)$ as the subgraph of $\bold{x}$ induced by $v_m$. For a graph $\bold{g}$ with $m$ vertices, we say  $\bold{x}(v_m)$ \emph{contains} $\bold{g}$ (and we write $\bold{x}(v_m) \succeq \bold{g}$) if $\bold{g}(i,j)=1$ implies $\bold{x}(v_m)(i,j)=1$. Then,  for $\bold{x} \in \mathcal{G}_N$ and $\bold{g} \in \mathcal{G}_m$, $m \leq N$, we define the number of subgraphs $\bold{g}$ in $\bold{x}$ by the counter

\begin{equation}\label{eq:count}
N_{\bold{g}}(\bold{x}) = \sum\limits_{v_m \in V_m}\mathds{1}\{\bold{x}(v_m) \succeq \bold{g}\}\,.
\end{equation}

Let $\bold{g}_1, \dots, \bold{g}_s$ be fixed graphs, where $\bold{g}_i$ has $m_i$ vertices, $m_i \leq N$ and $\Beta\in\R^s$ a vector of parameters. By convention we set $\bold{g}_1$ as the graph with only two vertices and one edge. Following \cite{chatterjee2013estimating} and \cite{bhamidi2011mixing} we define the probability 
of graph $\bold{x}$ by
\begin{equation}\label{eq:expo_beta}
p_N(\bold{x}|\Beta)=\dfrac{1}{Z_N(\Beta)}\exp\left( \sum\limits_{i=1}^s\beta_i\dfrac{N_{g_i}(\bold{x})}{N^{m_i-2}}\right) \,.
\end{equation}

Observe that the set $\mathcal{G}_N$ is equipped with a partial order given by
\begin{equation}\label{order}
 \gr{x} \preceq \gr{y}\quad  \text{if and only if}\quad \gr{x}(i,j)=1\quad\text{implies}\quad\gr{y}(i,j)=1\,.
\end{equation}
Moreover, this partial ordering has a maximal element $\gr{x}^{(1)}$ (the complete graph) and a  minimal element $\gr{x}^{(0)}$ (the empty graph).
Considering this partial order, we say  $p_N(\cdot\,|\Beta)$ is \emph{monotone} if the conditional distribution of $\bold{x}(i,j)=1$ given $\bold{x}_{-ij}$ is a monotone increasing function. That is, $p_N(\cdot\,|\Beta)$ is monotone if, and only if $\gr{x} \preceq \gr{y}$ implies
 
 \begin{equation}\label{mono1}
p_N(\gr{x}(i,j)=1|\Beta,\bold{x}_{-ij}) \;\leq\; p_N(\gr{y}(i,j)=1|\Beta,\bold{y}_{-ij})\,.
\end{equation}

Let $\mathbf{E}=\{(i,j): i,j \in V \mbox{ and } i< j\}$ be the set of possible edges in a graph with set of vertices V. For any graph $\gr{x} \in \mathcal{G}_N$, a pair of vertices $(i,j) \in \mathbf{E}$ and $a \in  \{0,1\}$  we define the modified graph $\gr{x}_{ij}^a \in \mathcal{G}_N$  given by 
\[ 
\gr{x}^a_{ij}(k,l) = \begin{cases}
x(k,l)\,, & \text{ if }(k,l)\neq (i,j)\,;\\
a \,, & \text{ if }(k,l) = (i,j)\,.
\end{cases}
\]
Then, Inequality \eqref{mono1} is equivalent to 

\begin{equation}\label{mono2}
\dfrac{p_N(\gr{x}_{ij}^{1}\, |\Beta)}{p_N(\gr{x}_{ij}^{0}\, |\Beta)}\; \leq\; \dfrac{p_N(\gr{y}_{ij}^{1}\, |\Beta)}{p_N(\gr{y}_{ij}^{0}\, |\Beta)}\;.
\end{equation}
Moreover, if this inequality holds for all  $\gr{x} \preceq \gr{y}$ then the distribution is monotone. 

Monotonicity plays a very important role in the development of efficient perfect simulation algorithms. 
The following proposition gives a sufficient condition on the parameter vector $\Beta$ under which the corresponding ERGM distribution is monotone.

\begin{prop}\label{prop:monotone}
Consider the ERGM  given by \eqref{eq:expo_beta}. 
If 
$\beta_i \geq 0$  for  all $i\geq 2$ then the distribution 
$p_N(\cdot\, ;\Beta)$ is monotone.
\end{prop}
The short proof of Proposition~\ref{prop:monotone} is given in Section~\ref{sec:proofs}.

\section{Coupling From The Past}

In this section we recall the construction of a Markov chain with stationary distribution $p_N(\cdot\, ;\Beta)$ using a local update algorithm called \emph{Glauber dynamics}. This construction and the monotonicity property given by Proposition~\ref{prop:monotone} are the keys to develop a perfect simulation algorithms for the ERGM.

\subsection{Glauber dynamics}
\label{sec:glauber}

In the Glauber dynamics the probability of a transition from graph $\gr{x}$ to  graph $\gr{x}^a_{ij}$ is given by 
\begin{equation}\label{eq:Glauber}
p(\gr{x}_{ij}=a| \Beta, \gr{x}_{-ij})=\dfrac{p_N(\gr{x}^a_{ij}|\Beta)}{p_N(\gr{x}^0_{ij}|\Beta)+p_N(\gr{x}^1_{ij}|\Beta)}
\end{equation}
and any other (\emph{non-local}) transition is forbidden. 

For the ERGM, this can be written
\begin{equation}\label{eq:Glauber_erg}
    p(\gr{x}_{ij}=1| \Beta, \gr{x}_{-ij}) = \frac{1}{1+ \exp \left\lbrace -2\beta_1+ \sum\limits_{k=2}^s \dfrac{\beta_k}{N^{m_k-2}} \left( N_{g_k}(\gr{x}^0_{ij})- N_{g_k}(\gr{x}^1_{ij})\right)     \right\rbrace}
 \end{equation}   
    and
    \[
       p(\gr{x}_{ij}=0| \Beta, \gr{x}_{-ij}) = 1 - p(\gr{x}_{ij}=1| \Beta, \gr{x}_{-ij})\,.
\]
It can be seen easilty that the Markov chain with the Glauber dynamics~\eqref{eq:Glauber_erg} has an invariant distribution equal to $p_N(\cdot\, ;\Beta)$. It can be implemented with the help of the local update function $\phi: \mathcal{G}_N\times \mathbf{E}\times[0,1] \rightarrow \mathcal{G}_N$ defined as

\begin{equation}\label{eq:update}
\phi(\gr{x},(i,j),u)=\begin{cases}
\gr{x}^{0}_{ij} , & \mbox{if  }  u \leq p(\gr{x}_{ij}=0| \Beta, \gr{x}_{-ij})\,;\\
\gr{x}^{1}_{ij}, &   \mbox{otherwise.} 
\end{cases}
\end{equation}
A transition in the Markov chain from graph $\gr{x}$ to graph $\gr{y}$ is obtained by the random function
\[
\gr{y} \,=\,\phi(\gr{x},\bold{e},\bold{u})\,,\quad \text{ with } \bold{e}\sim \text{Uniform}(\mathbf{E}),\; \bold{u} \sim \text{Uniform}(0,1)\,.
\]

For random vectors  $\bold{e}_{-n}^{-1}=(\mathrm{e}_{-n},\mathrm{e}_{-n+1},\dots,\mathrm{e}_{-1})$ and $\bold{u}_{-n}^{-1}=(u_{-n},u_{-n+1},$ $\dots,u_{-1})$, with $\mathrm{e}_{-k}\in \mathbf{E}$ and $u_{-k}\in [0,1]$, we define the random map  
$F_{-n}^0$
given by the following induction:
\begin{align*}
F_{-1}^0(\bold{x},\bold{e}_{-1}^{-1},\bold{u}_{-1}^{-1}) & = \phi(\gr{x},e_{-1},u_{-1})\,,\hbox{and}\\
F_{-n}^0(\bold{x},\bold{e}_{-n}^{-1},\bold{u}_{-n}^{-1}) &=F_{-n+1}^0(F_{-1}^0(\bold{x},\bold{e}_{-n}^{-n},\bold{u}_{-n}^{-n}),\bold{e}_{-n+1}^{-1},\bold{u}_{-n+1}^{-1})\,, \;\text{ for }n\geq 2\,. 
\end{align*}

The CFTP protocol of~\citep{propp1996exact} relies on the following elementary observation.
If, for some time $-n$ and for some uniformly distributed  $\bold{e}_{-n}^{-1}=(\mathrm{e}_{-1}\dots,\mathrm{e}_{-n})$ and $\bold{u}_{-n}^{-1}=(u_{-1},\dots,u_{-n})$, the mapping $F_{-n}^0(\cdot,\bold{e}_{-n}^{-1},\bold{u}_{-n}^{-1})$ is \emph{constant}, that is 
if it takes the same value at time $0$ on all possible graphs $\gr{x}$, then this constant value is readily seen to be a sample from the invariant measure of the Markov chain. Moreover, if $F_{-n}$ is constant for some positive $n$, then it is also constant for all larger values of $n$. 
Thus, it suffices to find a value of $n$ large enough so as to obtain a constant map, and to return the constant value of this map. For example, one may try some arbitrary value $n_1$, check if the map $F_{-n_1}^0(\cdot,\bold{e}_{-n_1}^{-1},\bold{u}_{-n_1}^{-1})$ is constant, try $n_2=2n_1$ otherwise, and so on...

Since the state space has the huge size $|\mathcal{G}_N|=2^{\frac{N(N-1)}{2}}$, it would be computationally intractable to compute the value of
$F_{-n}^0(\bold{x},\bold{e}_{-n}^{-1},\bold{u}_{-n}^{-1})$ for all $\bold{x} \in \mathcal{G}_N$ in order 
to determine if they coincide.  

This is where the monotonicity property helps: it is sufficient to inspect the  
 maximal and minimal elements of $\mathcal{G}_N$. If they  both yield the same value, then all other initial states will also coincide with them, and the mapping will be constant. 

To apply the monotone CFTP algorithm, we use the partial order defined by \eqref{order} and the extremal elements  $\gr{x}^{(0)}$ and  $\gr{x}^{(1)}$.
The resulting procedure, which simulates the ERGM using the CFTP protocol, is described in Algorithm 1.

\begin{algorithm}[th]
\caption{CFTP for ERGM}
\begin{algorithmic}[1]
\State Input: $\mathbf{E}$, $\phi$
\State Output: $\gr{x}$
\State $upper(i,j) \gets 1$, for all $(i,j)\in \mathbf{E}$
\State $lower(i,j) \gets 0$, for all $(i,j)\in \mathbf{E}$
\State $n \gets 0$
\WHILE{upper $\neq$ lower} 
\State $n \gets n+1$
\State $upper(i,j) \gets 1$, for all $(i,j)\in \mathbf{E}$
\State $lower(i,j) \gets 0$, for all $(i,j)\in \mathbf{E}$
\State Choose a pair of vertices $\mathrm{e}_{-n}$ randomly on $\mathbf{E}$
\State Simulate $u_{-n}$ with distribution uniform on $[0,1]$
\State upper $\gets F_{-n}^0 (upper,\bold{e}_{-n}^{-1} ,\bold{u}_{-n}^{-1})$
\State lower $\gets F_{-n}^0 (lower,\bold{e}_{-n}^{-1} ,\bold{u}_{-n}^{-1})$
\ENDWHILE
\State $\gr{x} \gets upper$\\
\Return Return: $n$, $\gr{x}$
\end{algorithmic}
\end{algorithm}

Let 
\[
T^{\text{stop}}_N=\min\{n >0 : F_{-n}^0(\gr{x}^{(0)},\bold{e}^{-1}_{-n},\bold{u}^{-1}_{-n})=F_{-n}^0(\gr{x}^{(1)},\bold{e}^{-1}_{-n},\bold{u}^{-1}_{-n})\}
\]
be the stopping time of  Algorithm 1. Proposition~\ref{the:dist_exp} below guarantees that the law of the output of Algorithm 1 is the target distribution  $p_N(\cdot;\Beta)$. The proof of Proposition~\ref{the:dist_exp}  follows directly from Theorem~2 in \citet{propp1996exact} and is omitted here.

\begin{prop}\label{the:dist_exp}
Suppose that  $\P(T^{\text{stop}}_N < \infty)=1$. Then the graph $\bold{x}$  returned by  Algorithm 1 has law $p_N(\cdot;\Beta)$.
\end{prop}

\section{Convergence Speed: Why CFTP is Better}
\label{sec:convergence_rates}

By the construction of the perfect simulation algorithm it is expected that the stopping time $T^{\text{stop}}_N$ is related to the mixing time of the chain. In Theorem~\ref{prop:CFTP} below we provide an upper bound on this quantity.

Given vectors $\bold{e}_{1}^{n}=(\mathrm{e}_{1},\mathrm{e}_{2},\dots,\mathrm{e}_{n})$ and $\bold{u}_{1}^{n}=(u_{1},u_{2}\dots,u_{n})$ define the forward random map $F_0^{n}(\gr{x},\bold{e}_{1}^{n},\bold{u}_{1}^{n})$ as
\begin{align}\label{forward-map}
F_{0}^1(\bold{x},\bold{e}_{1}^{1},\bold{u}_{1}^{1}) & = \phi(\gr{x},e_{1},u_{1})\,;\notag\\
F_{0}^n(\bold{x},\bold{e}_{1}^{n},\bold{u}_{1}^{n}) &=F_{0}^1(F_{0}^{n-1}(\bold{x},\bold{e}_{1}^{n-1},\bold{u}_{1}^{n-1}),\bold{e}_{n}^{n},\bold{u}_{n}^{n})\,, \;\text{ for }n\geq 2\,. 
\end{align}
For $x\in \mathcal{G}_N$, let the Markov chain $\{Y_n^{\bold{x}}\}_{n\in\mathbb N}$ be given by
\begin{align*}
Y_0^\bold{x} &= \bold{x}\\
Y_n^\bold{x} & = F_0^n (\bold{x},\bold{e}_{1}^{n},\bold{u}_{1}^{n})\,,\quad n\geq 1
\end{align*}
and denote by  $\mathrm{p}_n^{\bold{x}}$  the distribution of this Markov chain 
 at time $n$. Now define 
\[
\overline{d}(n)=\max\limits_{\bold{x},\bold{z} \in \mathcal{G}_N}\parallel \mathrm{p}_n^{\bold{x}}-\mathrm{p}_n^{\bold{z}} \parallel_{\text{TV}}\,,
\]
where $\|\cdot\|_{\text{TV}}$ stands for the total variation distance. 

Let $T_N^{\text{mix}}$ be given by 
\begin{equation}\label{tmix}
T_N^{\text{mix}}= \min\{\,n > 0 : \overline{d}(n) \leq \frac{1}{e}\, \}\,.
\end{equation}

The following theorem shows that the expected value of $T_N^\text{stop}$ is upper bounded by $T_N^\text{mix}$, up to a explicit multiplicative constant that depends on the number of vertices of the graph.

\begin{theorem}\label{prop:CFTP} Let $T_N^{\rm{mix}}$ be the mixing time defined by \eqref{tmix} for the forward map  
\eqref{forward-map}. Then 
\begin{equation}
\E[T^{\rm{stop}}_N] \;\leq\; 2\bigl( \log\bigl(  \tfrac{N(N-1)}{2}\bigr)+1\bigr)\,T_N^{\rm{mix}}\,.
\end{equation}
\end{theorem}

The proof of this relation is inspired from an argument given in~\cite{propp1996exact}, and is presented in Section~\ref{sec:proofs}.  
The upper bound in Theorem~\ref{prop:CFTP} 
could be also combined with other results 
on the mixing time of the chain to obtain an estimate of the expected value of $T^{\rm{stop}}_N$, when they are available.  As an example, we cite 
the results obtained by 

 \citet{bhamidi2011mixing}, where they present a study of the mixing time of the Markov chain constructed using the Glauber dynamics for ERGM
as we consider here. 

They show that for models where $\Beta$ belongs to the high temperature regime  the mixing time of the chain is $\Theta(N^2\log N)$. On the other hand, for models under the low temperature regime, the mixing time is exponentially slow.

Observe that the mixing time $T_N^{\rm{mix}}$ is directly related to the ``burn-in'' time in MCMC, where the non-stationary forward Markov chain defined by \eqref{forward-map} approaches the invariant distribution. This can be observed in practice in 
Figure~\ref{fig:plot_edges_stars_80}: the convergence time for both MCMC and CFTP are of the same order of magnitude, as suggested by Theorem~1. 
In this figure we compare some statistics 
of the graphs obtained by MCMC at different sample sizes of the forward Markov chain and by CFTP at its convergence time. The model is defined by \eqref{eq:expo_beta}, where $g_2$ is chosen as the graph with 3 vertices and 2 edges, with parameters 
$(\beta_1,\beta_2)=(-1.1,0.4)$,   
 and  the mean proportion of edges in the graphs is approximately $p=0.86$. 
  We observe in Figure~\ref{fig:01} that MCMC remains a long time strongly dependent on the initial value of the chain when this value is  chosen in a region of low probability. On the other hand, when the initial value is chosen appropriately, then MCMC and CFTP have comparable performances, as shown in Figure~\ref{fig:08}.

\begin{figure}[t!]
    \centering
    \begin{subfigure}[b]{1\textwidth}
    \centering
        \includegraphics[scale=0.5,trim=50 0 50 50]{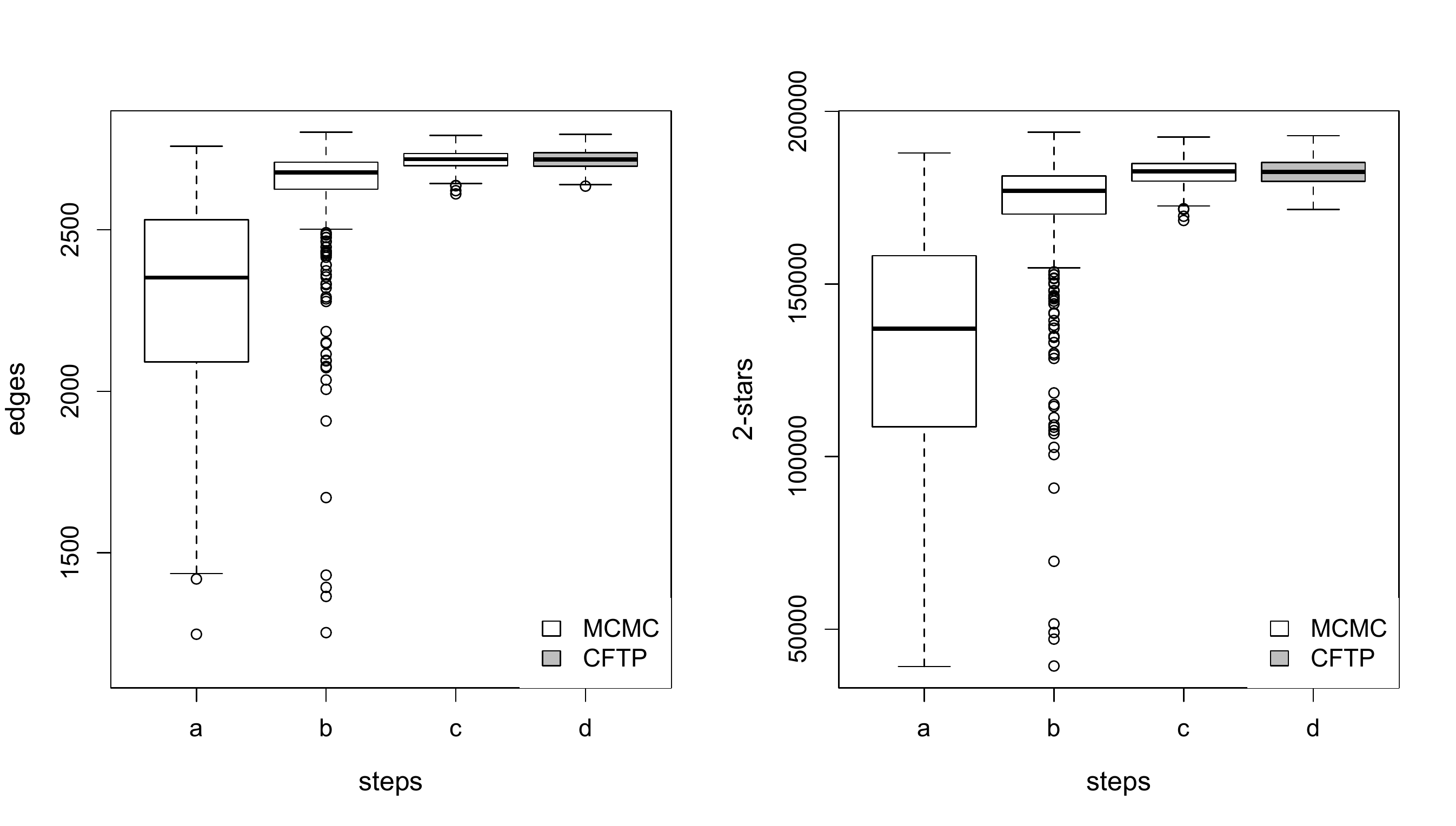}
        \caption{Initial state for MCMC chosen as Erd\H os-R\'enyi graph with parameter $0.1$.}
        \label{fig:01}
    \end{subfigure}
    \begin{subfigure}[b]{1\textwidth}
    \centering
    	\includegraphics[scale=0.5,trim=50 0 50 20]{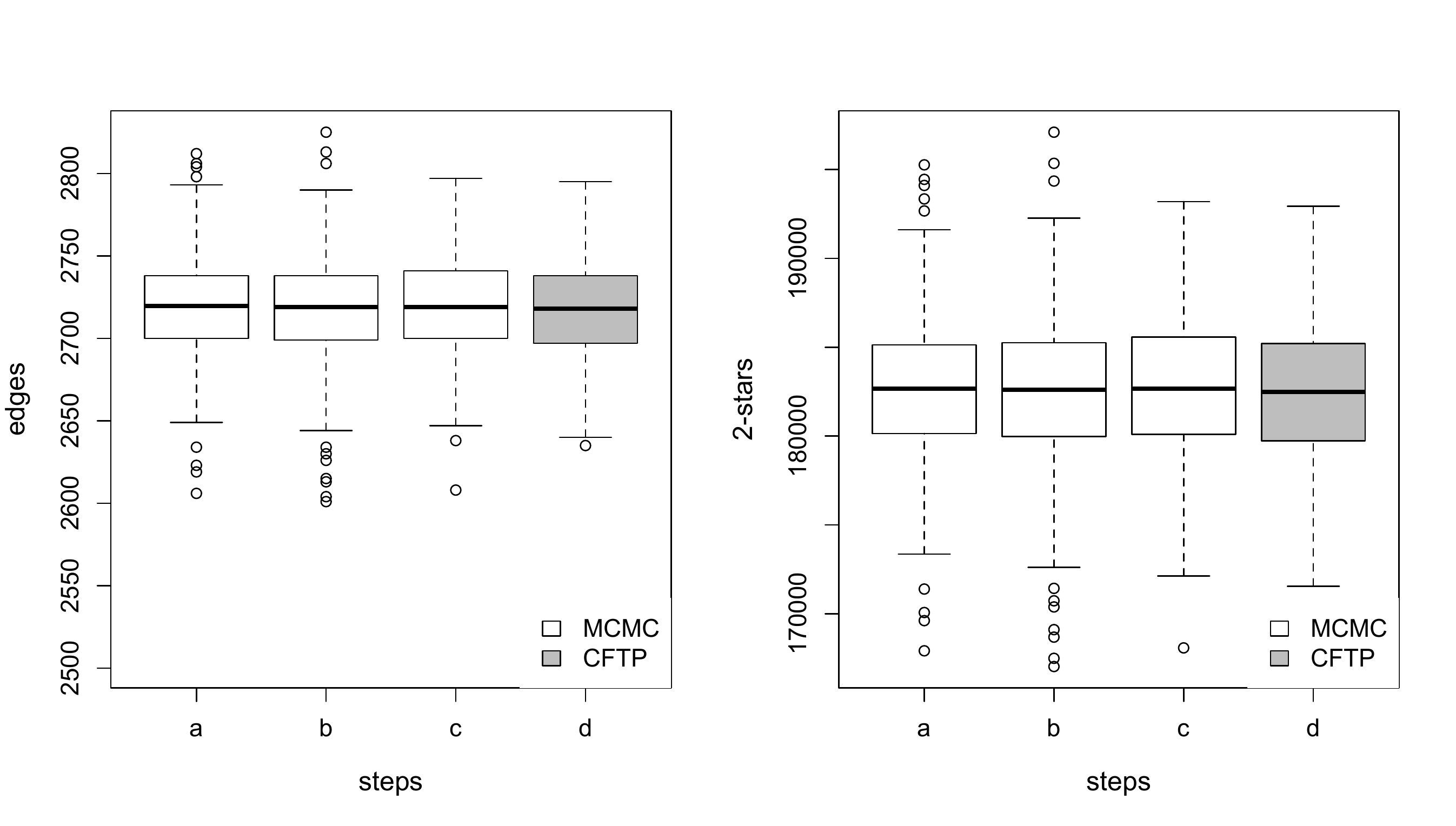}
        \caption{\small Initial state for MCMC chosen as Erd\H os-R\'enyi graph with parameter $0.8$.}
                \label{fig:08}
    \end{subfigure}
       \caption{Boxplots of number of edges (left) and number of 2-stars (right) for the ERGM with $N=80$ and parameter $(\beta_1,\beta_2)=(-1.1,0.4)$.  
    The MCMC algorithm was run for  $n= 80.000$ time steps (respectively $n=100.000$ (b) and $n=130.000$ (c)) and two different initial states. The results for the CFTP algorithm at the  time of convergence are shown in (d), with mean 130.500 time steps.}\label{fig:plot_edges_stars_80}
\end{figure}

\newpage
 
\section{Discussion}

In this paper we proposed a perfect simulation algorithm for the Exponential Random Graph Model, based on the
Coupling From The Past method of \citet{propp1996exact}.
The justification of the correctness of the CFTP algorithm is generally restricted to the monotonicity property of the dynamics, given by Proposition~\ref{prop:CFTP},  and to 
prove that $T_{N}^{\text{stop}}$ is almost-surely finite, but with no clue about the expected time to convergence. 
  Here, in contrast, we prove a much stronger result: not only do we upper-bound the expectation of the waiting time, but we show that CFTP compares favorably with standard MCMC forward simulation: the waiting time is (up to a logarithmic factor) not larger than the time one has to wait before a run of the chain approaches the stationary distribution. 
 Moreover, CFTP has the advantage of providing a draw from the exact stationary distribution, in contrast to MCMC which may keep for a long time a strong dependence on the initial state of the chain, as shown in the simulation example of Figure~\ref{fig:plot_edges_stars_80}.
We thus argue that for these models the CFTP algorithm is a better tool than forward MCMC simulation in order to get a sample from the Exponential Random Graph Model.

\section{Proofs}\label{sec:proofs}

\begin{proof}[Proof of Proposition \ref{prop:monotone}]
Let $\gr{x}$, $\gr{z} \in \mathcal{G}_N$ such that $\gr{x} \preceq \gr{z}$. 
For a pair of vertices $(i,j)$, we define a new count that considers only the subgraphs of $\bold{x}$ that contain the vertices $i$ and $j$, that is
\begin{equation}
N_{g_k}(\bold{x}, (i,j))=  \sum\limits_{\substack{v_{m_k} \in V_{m_k}\\ i,j \in v_{m_k}}}\mathds{1}\{\bold{x}(v_{m_k}) \succeq \bold{g}_k\}
\end{equation}
In the same way, we have that $N_{g_k}(\bold{z}_{ij}^{0}) \leq N_{g_k}(\bold{z}_{ij}^{1})$.  So, 
\begin{equation}
\begin{split}
0 \leq N_{g_k}(\bold{x}_{ij}^{1}) - N_{g_k}(\bold{x}_{ij}^{0}) = N_{g_k}(\bold{x}_{ij}^{1}, (i,j))\\
0 \leq N_{g_k}(\bold{z}_{ij}^{1}) - N_{g_k}(\bold{z}_{ij}^{0}) = N_{g_k}(\bold{z}_{ij}^{1}, (i,j))
\end{split}
\end{equation}
Since $\gr{x} \preceq \gr{z}$ we have that $N_{g_k}(\bold{x}_{ij}^{1}, (i,j)) \leq N_{g_k}(\bold{z}_{ij}^{1}, (i,j))$. For $\beta_k \geq 0$, $k=2,\cdots,s$, we have that
\begin{equation}
\begin{split}
& \sum\limits_{k=2}^s\dfrac{\beta_k}{N^{m_k-2}}\left(N_{g_k}(\bold{x}^{1}_{ij})-N_{g_k}(\bold{x}^{0}_{ij})- N_{g_k}(\bold{z}^{1}_{ij})+N_{g_k}(\bold{z}^{0}_{ij})\right) =\\
&  \sum\limits_{k=2}^s\dfrac{\beta_k}{N^{m_k-2}}\left( N_{g_k}(\bold{x}_{ij}^{1}, (i,j)) -  N_{g_k}(\bold{z}_{ij}^{1}, (i,j))  \right) \leq 0\\
\end{split}
\end{equation}
Finally,
\begin{align*}
\dfrac{p_N(\gr{y}_{ij}^{1}\, |\teta)}{p_N(\gr{y}_{ij}^{0}\, |\teta)} \dfrac{p_N(\gr{z}_{ij}^{0}\, |\teta)}{p_N(\gr{z}_{ij}^{1}\, |\teta)} \;&=\; \exp\left\lbrace \sum\limits_{k=2}^s\dfrac{\beta_k}{N^{m_k-2}}\left( N_{g_k}(\bold{x}_{ij}^{1}, (i,j)) -  N_{g_k}(\bold{z}_{ij}^{1}, (i,j))  \right) \right\rbrace\\
& \leq 1
\end{align*}
and this concludes the proof.
\end{proof}

\vspace{0.2cm}
\begin{proof}[Proof of Theorem \ref{prop:CFTP}]

Define the stopping time of the forward algorithm by 
\[
\widetilde{T}^{\text{stop}}_N=\min\{n >0 : F_{0}^n(\gr{x}^{(0)},\bold{e}^{n}_{1},\bold{u}^{n}_{1})=F_{0}^n(\gr{x}^{(1)},\bold{e}^{n}_{1},\bold{u}^{n}_{1})\}
\]
To prove the theorem claim we use the random variable $\widetilde{T}^{\text{stop}}_N$, since $\widetilde{T}^{\text{stop}}_N$ and ${T}^{\text{stop}}_N$ have the same probability distribution. In fact, 
\begin{equation}
\begin{split}
\P(T^{\rm{stop}}_N > n)&= \P(F_{-n}^0(\gr{x}^{(0)},\bold{e}_{-n}^{-1},\bold{u}_{-n}^{-1})\neq F_{-n}^0(\gr{x}^{(1)},\bold{e}_{-n}^{-1},\bold{u}_{-n}^{-1}))\\
&=\P(F_0^{n}(\gr{x}^{(0)},\bold{e}_{1}^{n},\bold{u}_{1}^{n})=F_0^{n}(\gr{x}^{(1)},\bold{e}_{1}^{n},\bold{u}_{1}^{n})=\P(\widetilde{T}^{\text{stop}}_N > n)\,.
\end{split}
\end{equation}

Let $\{Y^1_n\}_{n\in\mathbb N}$ and $\{Y^0_n\}_{n\in\mathbb N}$ be the Markov chains obtained by \eqref{forward-map} with initial states given by $\gr{x}^{(0)}$ and $\gr{x}^{(1)}$, respectively. Define $l(y)$ as the length of the longest increasing chain such that the top element is $y$. In the case $Y^0_n=Y^1_n$ we have $l(Y^0_n)=l(Y^1_n)$. Otherwise, if $Y^0_n\neq Y^1_n$ we have that $Y^0_n$ has at least one different edge of $Y^1_n$, since our algorithm use a local update. Then,  $l(Y^0_n) +1  \leq l(Y^1_n)$. So, we have that
\begin{equation}\label{eq:proof_prob}
\begin{split}
\P(\widetilde{T}^{\text{stop}}_N  >n) &= \P(Y^0_n\neq Y^1_n)=\P(l(Y^0_n) +1  \leq l(Y^1_n))\\
&\leq \E[l(Y^1_n) - l(Y^0_n)]=\vert \E_{\mathrm{p}_n^1} [l(Y)] - \E_{\mathrm{p}_n^0} [l(Y)] \vert\\
&\leq \parallel \mathrm{p}_n^1 - \mathrm{p}_n^0 \parallel_{\text{TV}}[\max\limits_{\bold{x}\in \mathcal{G}_N}l(\bold{x})-\min\limits_{\bold{x}\in \mathcal{G}_N}l(\bold{x})]  \leq \overline{d}(n)\max\limits_{\bold{x}\in \mathcal{G}_N}l(\bold{x})
\end{split}
\end{equation}
Since the update function of the transitions of the Markov chain given by \eqref{eq:update} updates only one edge at each step, we have that $\max\limits_{\bold{x}\in \mathcal{G}_N}l(\bold{x})=l(\bold{x}^{(1)})$ and $l(\bold{x}^{(1)})$ is the length of the longest increasing chain started on $\bold{x}^{(0)}$ with top element is $\bold{x}^{(1)}$, that is $l(\bold{x}^{(1)})=\frac{N(N-1)}{2}$. Thus
\[
\overline{d}(n) \;\geq\; \frac{2}{N(N-1)}\,\P(\tau^{*}_n>n)\,.
\]
Since $\P(\widetilde{T}^{\text{stop}}_N  >n)$ is submultiplicative (Theorem 6 in \cite{propp1996exact}) we have that
\begin{equation}\label{eq:proof_exp}
\E[\widetilde{T}^{\text{stop}}_N ]\leq \sum\limits_{i=1}^\infty n\P(\widetilde{T}^{\text{stop}}_N  > in) \leq \sum\limits_{i=1}^\infty n[\P(\widetilde{T}^{\text{stop}}_N  > n)]^i= \frac{n}{\P(\widetilde{T}^{\text{stop}}_N \leq n)}\,.
\end{equation}
Then, \eqref{eq:proof_prob} and \eqref{eq:proof_exp} imply that
\begin{equation}\label{eq:proof_exp1}
\E[\widetilde{T}^{\text{stop}}_N ] \;\leq\; \frac{n}{1-\overline{d}(n)\frac{N(N-1)}{2}}\,.
\end{equation}
Set $n=\left( \log( \frac{N(N-1)}{2})+1\right)T^{\rm mix}_N$. Using the submultiplicative property of $\overline{d}(n)$ (Lemma 4.12 in \cite{levin2009markov}) we
have that
\begin{equation}\label{eq:proof_exp3}
\overline{d}(n)\;\leq\; [\overline{d}(T^{\rm mix}_N)]^{ \log( \frac{N(N-1)}{2})+1}\; \leq\; \frac{2}{eN(N-1)}
\end{equation}
and by \eqref{eq:proof_exp1} and \eqref{eq:proof_exp3} we conclude that
\begin{equation*}\label{eq:proof_exp2}
\E[\widetilde{T}^{\text{stop}}_N ]\;\leq\;  \dfrac{\left( \log\left(  \frac{N(N-1)}{2}\right)+1\right)T^{\rm mix}_N }{1-\frac{1}{e}}\; \leq\; 2\left( \log\left(  \tfrac{N(N-1)}{2}\right)+1\right)T^{\rm mix}_N \,.\qedhere
\end{equation*}
\end{proof}

\section*{Acknowledgments}
A.C. is supported by a FAPESP scholarship (2015/12595-4).
F.L. is partially supported by a CNPq' fellowship (309964/2016-4). 
 This work
  was produced as part of the activities of FAPESP Research, Innovation and Dissemination Center for Neuromathematics, grant 2013/07699-0, and FAPESP's project  \emph{Structure selection for stochastic processes in high dimensions}, grant 2016/17394-0, S\~ao Paulo Research Foundation.

\bibliography{./references}

\end{document}